\documentclass[letterpaper,10pt,conference]{ieeeconf}
\IEEEoverridecommandlockouts      
 \author{Shuang Gao 
 \thanks{Shuang Gao is with the Simons Institute for the Theory of Computing at University of California, Berkeley, CA, USA, 94720. 
Email: {\tt\small shuang.gao@berkeley.edu,~sgao@cim.mcgill.ca}.  
}%
  \thanks{The author would like to thank Peter E. Caines and  Minyi Huang for valuable discussions and feedback on this work, and thank Matthew O. Jackson for conversations about several important references related to this work and for sharing his insights.}
    \thanks{The author gratefully acknowledge the supported by the U.S. ARL and ARO grant W911NF1910110,  the U.S. AFOSR grant FA9550-19-1-0138, and the Simons-Berkeley Research Fellowship.}
  }

\usepackage[noadjust,sort]{cite}

\usepackage{graphics}
\usepackage{epstopdf}

\usepackage{amsmath}
  \usepackage[font=footnotesize]{subfig}

\usepackage{enumerate}
\usepackage{color}
\usepackage{booktabs}

\newcommand{\FA}{\mathbf{A}}   
\newcommand{\FB}{\mathbf{B}}   
\newcommand{\FD}{\mathbf{D}}   
\newcommand{\Fx}{\mathbf{x}}	   
\newcommand{\Fv}{\mathbf{v}}	   

\newcommand{\FW}{\mathbf{W}}

\newcommand{\BR}{\mathds{R}}  

\newcommand{\ESZ}{\mathcal{W}_0}
\newcommand{\ESO}{\mathcal{W}_1}
\newcommand{\ESC}{\mathcal{W}_c}

\newcommand*\TRANS{{\mathpalette\doTRANS\empty}}
\makeatletter
\newcommand*\doTRANS[2]{\raisebox{\depth}{$\m@th#1\intercal$}}
\makeatother

\usepackage{url}
\usepackage{amssymb,mathtools}

\usepackage[standard,thmmarks,amsmath]{ntheorem}
\usepackage{times}
\usepackage{dsfont}
\usepackage{tikz}

\usepackage{hyperref}

\begin{document}
%
\title{Fixed-Point Centrality for Networks} 
%
%
%
%
%
\maketitle 
\begin{abstract}
This paper proposes a family of network centralities called fixed-point centralities. This centrality family is defined via the fixed point of permutation equivariant mappings related to the underlying network. Such a centrality notion is immediately extended to define fixed-point centralities for infinite graphs characterized by graphons.  Variation bounds of such centralities with respect to the variations of the underlying graphs and graphons under mild assumptions are established.    
Fixed-point centralities connect with a variety of different models on networks including graph neural networks, static and dynamic games on networks, and Markov decision processes.

\end{abstract}

\section{Introduction}
Centrality which quantifies the ``importance'' or ``influence'' of nodes on  networks is a useful concept in social network analysis~\cite{faust1997centrality,landherr2010critical,jackson2010social} %
 and it also finds applications in biological, technological and economics networks (see e.g. \cite{koschutzki2008centrality,wang2010electrical,ballester2006s,jackson2010social}). 
 Plenty of centralities with different properties are defined for different problems (see e.g. \cite{newman2018networks,bloch2016centrality}), such as, degree centrality, eigencentrality \cite{bonacich1972technique},  Katz-Bonacich centrality \cite{katz1953new,bonacich1987power,bonacich2001eigenvector}, PageRank centrality \cite{brin1998anatomy},  Shapley value \cite{shapley1953value}, closeness centrality \cite{bavelas1950communication}, betweenness centrality \cite{freeman1977set}, diffusion centrality \cite{banerjee2013diffusion},  among others.
  These centralities provide  a collection of different quantitative measures for the ``importance'' or ``influence'' of nodes on  networks associated with various application contexts.
For instance,   the quality of a website may be modelled by the PageRank centrality \cite{brin1998anatomy},   the importance of individuals on social influence networks may be reflected by eigencentrality in \cite{bonacich1972technique}, equilibrium actions of certain static network games correspond to Katz-Bonacich centrality \cite{ballester2006s},  contribution values in a coalition game may be represented by Shapley values \cite{shapley1953value}, and so on.
Many social, technological and biological networks are growing and varying in terms of nodes and (or)  connections and hence centrality values may vary accordingly.  
Properties of such variations of centrality values with respect to the variations of graphs (see \cite{avella2018centrality}) motivate the current work.  %
A second motivation is to identify a suitable centrality notion for dynamic game problems on networks and graphons (\cite{bacsar1998dynamic,ShuangPeterMinyiCDC21,ShuangPeterMinyiTAC21,PeterMinyiSIAM21GMFG}). A third motivation  is the search of a class of new centralities for centrality-weighted opinion dynamics models proposed in \cite{ShuangCDC21}.

\subsection{Related work} 
The formulation of fixed-point centrality in this paper follows the idea of the seminal work on graph neural network models (\cite{gori2005new,scarselli2009graph}) in using  fixed points of some underlying mappings associated with networks. 
Fixed-point characterizations find  applications in many problems in data science, including graph neural networks (\cite{gori2005new,scarselli2009graph}),  implicit neural networks (\cite{gu2020implicit,el2021implicit,jafarpour2021robust,davydov2021non}), deep equilibrium models \cite{bai2019deep}, among others
 \cite{combettes2021fixed}. 
A first salient feature of fixed-point centralities that distinguishes themselves from these models above is that permutation equivariance properties must be satisfied.
Another salient feature of fixed-point centralities is that the values of fixed-point centralities are restricted to real numbers to allow natural rankings of the nodes and are restricted to non-negative numbers to allow interpretations (after normalizations) as  probability distributions. 
Furthermore, the current paper focuses on variations of the fixed point (centralities)  with respect to the variations of graph structures and weights, which differs from  \cite{gori2005new,scarselli2009graph,gu2020implicit,el2021implicit,jafarpour2021robust,davydov2021non,bai2019deep}.

Centralities and graph neural networks are respectively generalized in \cite{avella2018centrality} and  \cite{ruiz2020graphon}  to those for infinite graphs characterized by graphons (developed in \cite{lovasz2006limits,borgs2008convergent,borgs2012convergent}  to characterize dense graph sequences and  their limits).
The work \cite{avella2018centrality} studies the eigencentrality, PageRank centrality, Katz-Bonacich centrality of symmetric graphs generated from  graphons and establishes the rate of convergence of these centralities to the associated graphon centralities. 
The fixed-point centrality for graphon  in the current paper provides a unified view towards these centralities. 
The graphon versions of graph neural networks as approximations or generalization models of graph neural networks are proposed and analyzed in \cite{ruiz2020graphon}. One modelling difference is that the graphon neural networks are characterized by  layered structures in \cite{ruiz2020graphon} whereas in the current paper fixed-point equilibrium structures are employed.
\subsection{Contribution}
We propose the ``\emph{fixed-point centrality}'', which is a class of centralities that can be constructed via a (permutation equivariant) fixed-point mapping associated with the underlying graph. 
This class of centralities unifies many different centralities (including PageRank centrality, eigencentrality, and Katz-Bonacich centrality) and furthermore it connects to a variety of different problems including graph neural networks \cite{gori2005new}, and LQG mean field games on networks \cite{ShuangPeterMinyiCDC21}. 
In addition, fixed-point centralities are applicable to a broader class of graphs, whether they are undirected or directed, unweighted or weighted (with possibly negative weights), finite or infinite.
Moreover,  variation bounds of fixed-point centralities with respect to the variations of the underlying graphs are established  under mild assumptions following a rather simple idea based on fixed-point analysis.

\emph{Notation}:
   $\BR$ and $\BR_{\geq 0}$ denote respectively reals  and non-negative reals. For $A \in \BR^{n\times n}$, $\mathcal{G}(A)$ denotes the graph with the adjacency matrix $A$ and the node set $[n]\triangleq \{1,...,n\}$. $\mathcal{G}(V,E)$ denotes the graph with vertex set $V$ and edge set $E\subset V\times V$.  For a vector $v\in \BR^n$, $\text{span}(v)\triangleq\{\alpha v: \alpha \in \BR^n\}$. $\mathbf{1}_n$ denotes the $n$-dimensional column vector of $1$s and $\mathbf{1}$ denotes the function defined over $[0,1]$ with $\mathbf{1}_\alpha =1$ for all $\alpha\in [0,1]$. We use the word ``network'' to refer to an interconnected group or system where the connection structures along with weights can be characterized by some graph $\mathcal{G}(A)$. For a vector $v\in \BR^n$, $\text{diag}(v)$ denotes the $n\times n$ diagonal matrix with the elements of $v$ on the main diagonal, $[v]_i$ (or $v_i$) denotes the $i$th element of $v$,   
    $ \left\|{v} \right\|_{p}\triangleq\left(\sum _{i=1}^{n}\left|v_{i}\right|^{p}\right)^{1/p}$ with $1\leq p< \infty$, \textup{and} $\|v\|_\infty  \triangleq \max_{i\in [n]} |v_i|. $ For a matrix $A\in\BR^{n\times n}$,   $[A]_{ij}$ (or $a_{ij}$) denotes the $ij$th element of  $A$ and  
    $\|A\|_{p}\triangleq \sup_{v\in \BR^n, v\neq 0} \frac{\|Av\|_p}{\|v\|_p}$ with  $1\leq p\leq \infty$. We note that  $\|A\|_{1}=\max _{1\leq j\leq n}\sum _{i=1}^{m}|a_{ij}| ~  \text{and} ~   \|A\|_{2} ={\sqrt {\lambda _{\max }\left(A^{*}A\right)}}.$  
\section{Preliminaries on Centralities}
	A \emph{centrality} for a network characterized by a graph $\mathcal{G}(V,E)$ is  a mapping $\rho : V \to \BR_{\geq 0}$ that provides a quantification of ``importance'' or ``influence'' of nodes on the network. 
It is worth emphasizing that the ``importance" or ``influence'' of nodes on networks is defined differently under different application contexts. Hence for the same graph structure and graph weights, various centralities can be defined and may be very different from one another.  
The choice of the range $\BR_{\geq 0}$ from centralities  allows a natural ranking of nodes.
{The fundamental idea of centralities is to summarize the  information about a two-variable function characterized by a graph (or a matrix) into a one-variable function characterized by a centrality (or a vector).} 

 We review several centralities related to the current paper. %
\subsection{Centralities for Finite Networks}
Consider a  graph $\mathcal{G}(A)$ with non-negative adjacency matrix $A=[a_{ij}] \in \BR^{n\times n}$. Depending on $A$, the graph $\mathcal{G}(A)$ may be directed or undirected, and  weighted or unweighted. 
\begin{enumerate}[\bf (E1)]

\item{Eigencentrality} (proposed in \cite{bonacich1972technique}):
Assume the largest eigenvalue $\lambda_1$ of $A$ is simple (i.e.  $\lambda_1$ has multiplicity~$1$).  Then the eigencentrality of $\mathcal{G}(A)$ is given by
\[
\rho_i =  \left[\lim_{k\to \infty}\Big(\frac1{\lambda_1} A^\TRANS\Big)^k \mathbf{1}_n
\right]_i,  \quad   i \in [n].
\] An equivalent form in terms of local connections is  
\[\rho_{i}=\frac1{\lambda_1} \sum _{j=1}^{N}a_{ji} \rho_{j},~~ i\in[n],\quad \text{i.e.}\quad  \rho = \frac1{\lambda_1} A^\TRANS \rho.\]
\item {Katz-Bonacich centrality} with $\alpha \in (0,1)$ (proposed in \cite{katz1953new} and generalized in \cite{bonacich1987power,bonacich2001eigenvector}):
Let  $\alpha <\|A\|_2^{-1}$. One (simplest) Katz-Bonacich centrality is given by
\begin{equation*}
	\begin{aligned}
\rho_i &= \sum_{k=0}^{\infty}\sum_{j=1}^n \alpha^k [A^k]_{ji}
	 = \left[ \sum_{k=0}^{\infty} \alpha^kA^{\TRANS ^k}  \mathbf{1}_n \right]_i, \quad i \in [n],
\end{aligned}
\end{equation*}
 where the upper bound of $\alpha$ ensures the boundedness of the infinite series.
An equivalent form in terms of local connections is given by 
\[\rho_{i}=\alpha \sum _{j=1}^{n}a_{ji}\rho_{j} +1, ~ i \in [n], \quad \text{i.e.}\quad \rho  = \alpha A^\TRANS \rho + \mathbf{1}_n,
\]
and the equivalent explicit form is  
$
\rho  =  (1-\alpha A^\TRANS)^{-1} \mathbf{1}_n. 
$
%
\item {PageRank centrality} (proposed in \cite{brin1998anatomy}): Consider a network of webpages, where each node represents a webpage, and  $a_{ji}=1$ if there is a hyperlink from webpage $j$ to $i$ and $a_{ij}=0$ otherwise \cite{brin1998anatomy}.
  PageRank centrality with $\alpha \in(0,1)$ is given by
\[
\rho_{i}=\alpha \sum_{j=1}^na_{ji}{\frac {\rho_{j}}{d_j}}+{\frac {1-\alpha }{n}},\quad d_j=\sum _{i=1}^na_{ji},  \quad i \in [n],
\]
where  \textcolor{black}{$\alpha \frac{a_{ji}}{{d_j}}\rho_j+ (1-\alpha)n^{-2}$  is the probability of jumping from node $j$ to node $i$} in the steady state of the associated random walk.  
 In equivalent forms, PageRank centrality $\rho$ satisfies
\[
\rho = \alpha A^\TRANS D^{-1}  \rho + \frac{1-\alpha}{n}\mathbf{1}_n, ~~ \text{with}~ D\triangleq\textup{diag}(d_1,...,d_n)
\]
and the explicit computation form is  given by 
\[
\rho  = \frac{1-\alpha}{n}(I-\alpha A^\TRANS D^{-1})^{-1} \mathbf{1}_n.
\]
  PageRank centrality can be interpreted as the steady state distribution of random walks on the network. %
%
\end{enumerate}
\subsection{Centralities for Graphons}

\emph{Graphons} are defined as bounded symmetric measurable functions $\FA: [0,1]^2 \to [0,1]$, which, roughly speaking, can be  viewed as the ``adjacency matrix" of  graphs with the vertex set $[0,1]$  (see \cite{lovasz2012large}). Let $\ESZ$ denote the set of graphons with the range $[0,1]$. A graphon $\FA \in \ESZ$ can be interpreted as an integral operator (for instance from $L^2([0,1])$ to $L^2([0,1])$) as follows: 
 \[
 (\FA \Fv)(\cdot) = \int_{[0,1]} \FA(\cdot, \alpha)\Fv_\alpha d\alpha, \quad \Fv \in L^2([0,1]).
 \]
The definitions of eigenvector, PageRank and Katz-Bonacich centralities for graphons  in \cite{avella2018centrality} are summarized below. Consider a graphon $\FA \in \ESZ$. 
\begin{enumerate}
	\item [\bf(E4)] The graphon eigencentrality for $\FA$ is given by 
\[
\rho = \frac1{\lambda_1} \FA \rho, \quad (\FA\rho)(\cdot) \triangleq \int_{[0,1]}\FA(\cdot, \alpha) \rho_\alpha d\alpha,
\]
where $\rho$ denotes the eigenfunction in $L^2([0,1])$ associated to the largest eigenvalue $\lambda_1$ of $\FA$ and $\lambda_1$ is assumed to have multiplicity $1$. 
\item[\bf (E5)] The graphon Katz-Bonacich centrality with $\alpha \in (0,1)$  for $\FA$ is defined by one of the  equivalent forms:
\begin{equation*}
	\begin{aligned}
\rho  = \sum_{k=1}^{\infty}\alpha^k  \FA^k\mathbf{1},   ~ \rho =  (I - \alpha \FA)^{-1} \mathbf{1}, \text{or }\rho  =   ~\alpha \FA \rho + \mathbf{1}
\end{aligned}
\end{equation*}
where $\alpha < \frac{1}{\lambda_{1}}$ and $\lambda_{1}$ is the largest eigenvalue of $\FA$. 
\item[\bf (E6)] The graphon {PageRank centrality} with $\alpha \in (0,1)$ for  $\FA$ is given by
\begin{equation}\label{eq:page-rank-graphon}
	\rho = \alpha \FA \odot \FD^{-1} \rho + ({1-\alpha})\mathbf{1}, \quad \FA \in \ESZ,  
\end{equation}
where $\FD(x) = \int_{[0,1]} \FA(y,x)dy, $
and  $(\FA\odot \FD^{-1})(x)= \frac{\FA(x,y)}{\FD(y)}$ if $\FD(y)\neq 0$,  and zero otherwise. 
Equivalent representation forms are as follows:
$
\rho = (1-\alpha)(I-\alpha \FA\odot \FD^{-1})^{-1}\mathbf{1}
$ and
$
\rho = (1-\alpha)\sum_{k=0}^\infty \big(\alpha \FA \odot \FD^{-1}\big)^k\mathbf{1}.
$
\end{enumerate}
\begin{proposition}
The graphon PageRank centrality $\rho$  with $\alpha \in(0,1)$ is a probability density function over $[0,1]$. 	
\end{proposition}
\begin{proof}
From the equivalent form $
\rho = (1-\alpha)\sum_{k=0}^\infty \big(\alpha \FA \odot \FD^{-1}\big)^k\mathbf{1}
$, we obtain that $\rho(x)\geq 0$ for all $x\in[0,1]$ for $\FA \in \ESZ$. Furthermore, based on \eqref{eq:page-rank-graphon}, we verify that
$$
 \langle \mathbf{1}, \rho\rangle = \langle \mathbf{1}, \alpha  \FA \odot \FD^{-1} \rho \rangle   + (1-\alpha)\langle \mathbf{1}, \mathbf{1}\rangle  =1. 
$$
 Thus $\rho$ is a probability density.
\end{proof}
\section{Fixed-Point Centrality for Finite Networks}
A \emph{permutation matrix} is a square matrix that has exactly one element of $1$ in every row and every column and $0$s elsewhere. An $n\times n$-dimensional permutation matrix $P_\pi$
 can be obtained by permuting the rows of an $n\times n$ identity matrix according to the permutation map $\pi:[n]\to [n]$.   For any permutation map $\pi:[n]\to [n]$, its associated permutation matrix $P_\pi$ is orthonormal, that is, $P_\pi^\TRANS P_\pi = I$. 

\begin{definition}[Permutation Equivariance]
A mapping $f(\cdot, \cdot): \BR^{n\times n} \times \BR^n \to \BR^n$ is \emph{permutation equivariant with respect to a permutation map $\pi:[n]\to [n]$} if 
$$
P_\pi f(A , \rho)  = f(P_\pi A P^\TRANS_\pi, P_\pi \rho), \quad \forall \rho \in \BR^n, ~{\forall A \in \BR^{n\times n}}
$$
where  $P_\pi $ is the permutation matrix corresponding to $\pi$. 
	A mapping $f(\cdot, \cdot): \BR^{n\times n} \times \BR^n \to \BR^n$ is \emph{permutation equivariant} if it is {permutation equivariant} with respect to all permutation maps $\pi:[n]\to [n]$.
\end{definition}

\begin{definition}[Permutation Invariance]
	A mapping $f(\cdot, \cdot): \BR^{n\times n} \times \BR^n \to \BR^n$ is \emph{permutation invariant with respect to permutation map $\pi:[n]\to [n]$} if
   $$
f(A , \rho)  = f(P_\pi A P^\TRANS_\pi, P_\pi \rho),  ~~{\forall \rho \in \BR^n, ~\forall A \in \BR^{n\times n}}
$$
where  $P_\pi $ is the permutation matrix corresponding to $\pi$. 
	A mapping $f(\cdot, \cdot): \BR^{n\times n} \times \BR^n \to \BR^n$ is \emph{permutation invariant}  
if it is permutation invariant with respect to all permutation maps $\pi:[n]\to [n]$.
\end{definition}
Similarly, a mapping $g(\cdot):\BR^n \to \BR^n$ is \emph{permutation equivariant} (resp. \emph{permutation invariant})  if $P_\pi g(\rho) = g(P_\pi \rho)$ (resp. $g(\rho) = g(P_\pi \rho)$) for all $\rho \in \BR^n$ and all permutation maps $\pi:[n]\to [n]$. 

Permutation  equivalence  and permutation invariance are important properties of many functions associated with DeepSets \cite{zaheer2017deep} and  graph neural networks  \cite{bronstein2021geometric}. 
{Consider a network, the structure of which is characterized by a graph $\mathcal{G}(A)$ with the adjacency matrix $A \in \BR^{n\times n}$ (which may have negative weights) and the node set $[n]$.  $S$ denotes a set of nodal features and $S^n$ denotes its $n$-fold Cartesian  product.  Let $S^n$ be associated with a metric~$d$.}
\begin{definition}[Fixed-Point Centrality] \label{def:fpc-net}
	A centrality $\rho:[n]\to \BR_{\geq 0}$  is a \emph{fixed-point centrality} for $\mathcal{G}(A)$ associated with the feature space $(S^n, d)$ if there exists a permutation equivariant  mapping $f(\cdot, \cdot): \BR^{n\times n} \times S^n \to S^n$,  a permutation equivariant mapping $g(\cdot):  S^n \to \BR^n_{\geq0}$, and  \textcolor{black}{a unique $x \in S^n$ under the metric $d$} such that 
\begin{equation}\label{eq:fpc-def}
\begin{aligned}
	&x = f(A, x),  ~~ x \in S^n,  \\
	& \rho = g(x), ~~~~ \rho \in \BR^n_{\geq 0}.
\end{aligned}
\end{equation}
\end{definition}
We note the (symmetric or asymmetric) adjacency matrix $A=[a_{ij}]$ of $\mathcal{G}(A)$ is allowed to have non-negative elements. 
The choices of $f$ and $g$ are contingent to the network application context and hence different fixed-point centralities may be associated with the same underlying graph $\mathcal{G}(A)$. 

The existence of the fixed-point feature is assumed in the definition of fixed-point centrality, which, 
with extra assumptions, can be established via various fixed-point theorems \cite{agarwal2001fixed} (see, for instance, \cite{lyu2021centrality} based on Kakutani fixed-point theorem and \cite{scarselli2009graph} based on Banach fixed-point theorem). The uniqueness of the fixed-point feature depends on the properties of both $A$ and $f$, as it is determined by $f(A,\cdot): S^n \to S^n$.  
	 Thus, for the same permutation equivariant mapping $f(\cdot, \cdot)$,
	a  different $A$ may result in the non-uniqueness (or even non-existence) of the fixed-point feature. 
	 To enforce uniqueness of the fixed-point feature (when it exists), one way is to select a suitable feature set $S$ along with the metric $d$ for the product space $S^n$ such that uniqueness is defined up to equivalent classes (see Remark \ref{eq:linear-case} below for an example).

\begin{remark}[Linear Case]\label{eq:linear-case}
	When $S$ is a vector space and $f(A,\cdot)$ is a linear function from $S^n$ to $S^n$  (that is,  $f(A, x_1+ x_2) = f(A, x_1)+ f(A,x_2)$  and $f(A, \alpha x_1) = \alpha f(A, x_1)$ for any $x_1, x_2 \in S^n$, $\alpha \in \BR$) for any $A \in \BR^{n\times n}$,  the unique $x \in S^n$ in \eqref{eq:fpc-def} should be interpreted as the unique $1$-dimensional subspace, or in other words, $x \in S^n$ is unique up to its linear span;  a formal way to have the uniqueness is to extend the feature vector space $S^n$ to the Grassmannian $Gr(1, S^n)$ (i.e. the space of $1$-dimensional linear subspaces in $S^n$) and use a distance for $Gr(1, S^n)$ (see e.g. \cite{ye2016schubert}). 
\end{remark}

Fixed-point centralities can be viewed as a specialization of  graph neural network models in \cite{gori2005new,scarselli2009graph} to the case where outputs are characterized by non-negative reals and {initial nodal labels there are homogenous}. 
Centralities in $\BR_{\geq 0}$ naturally allow ranking nodes  according to their centrality values, whereas in general  outputs of graph neural networks require extra constructions to allow such ranking.
\subsection{Examples of Fixed-Point Centrality}
\begin{proposition}
Eigencentrality, Katz-Bonacich centrality, \textcolor{black}{and} PageRank centrality are  fixed-point centralities. 
\end{proposition}
\begin{proof}
The proof is by identifying the functions $f$ and $g$  following the definition of fixed-point centrality. 
The mapping $g$ in \eqref{eq:fpc-def} is specialized to the identify mapping from $\BR^n \to \BR^n$ (i.e. $\rho =x$) for \textcolor{black}{Katz-Bonacich centrality}, PageRank centrality and  eigencentrality.
 For Katz-Bonacich centrality, $f(A,x) = \alpha A^\TRANS x+ \mathbf{1}_n, ~ \alpha \in (0,~\|A\|_2^{-1}).$
  We observe that $\|\alpha A^\TRANS\|_2<1$ and hence $f(A,\cdot)$ for Katz-Bonacich centrality is a contraction \textcolor{black}{from} $\BR^n$ to $\BR^n$ under the vector 2-norm. 
 For PageRank centrality, $\alpha \in (0,1)$, 
 $f(A,x) = \alpha  A^\TRANS \text{diag}(A^\TRANS\mathbf{1}_n)^{-1} x + \frac{1-\alpha}{n}\mathbf{1}_n. $
 We note that $\|\alpha A^\TRANS \text{diag}(A^\TRANS\mathbf{1}_n)^{-1}\|_1=\alpha <1$ and hence $f(A,\cdot)$ for  PageRank centrality is a contraction under vector 1-norm from $\BR^n$ to $\BR^n$. 
The existence and uniqueness of  fixed-point features for these two cases above are immediate via Banach fixed-point theorem. 
 For the case with eigencentrality,  the largest eigenvalue $\lambda_1$ of $A$ is assumed to have multiplicity $1$, and the permutation equivariant mapping is 
 $f(A,x) = \frac1{\lambda_{1}}Ax$.
  The fixed-point feature $x$ is unique up to its linear span as  $f(A,\cdot)$ is a linear function from $\BR^n$ to $\BR^n$ (see Remark \ref{eq:linear-case}).
  The permutation equivariance properties of  functions $f(\cdot,
 \cdot)$ for these centralities  can be easily verified.  
\end{proof}

\begin{proposition}\label{prop:eigenvec}
Any eigenvector corresponding to a nonzero simple eigenvalue of $A$ is a fixed-point centrality for $\mathcal{G}(A)$.	
\end{proposition}
Proofs are omitted as readers can readily verify the result.

We emphasize that the choice of $S$ in \eqref{eq:fpc-def} can be very general: it can be a set of vectors, matrices, functions, probability distributions,  strings, etc.  Below we give an example where $S$ is the space of continuous functions from $[0,T]$ to $\BR^q$ denoted by $C([0,T]; \BR^q)$ with $q\geq 1$.
%
\begin{proposition}
The equilibrium nodal cost of LQG Network Mean Field Games \textup{\cite[Sec. IV-B]{ShuangPeterMinyiTAC21}} with homogenous initial conditions, if the unique equilibrium exists, is a fixed-point centrality. 
\end{proposition}
\begin{proof}
Following \cite[Prop. 1]{ShuangPeterMinyiTAC21}, the network mean field trajectory denoted by $z= (z_1,...., z_n)^\TRANS$ with $z_i(t) \in \BR^{q}$ for $t\in [0,T]$ on a network with $n$ nodes satisfies  
\[
z  = \Phi (A, z), \quad z \in (C([0,T]; \BR^{q}) )^n,
\]
where  $C([0,T]; \BR^{q})$ denotes the space of continuous functions from $[0,T]$ to $\BR^{q}$  (endowed with the sup norm) and the permutation equivariant  mapping  $$\Phi(\cdot,\cdot): \BR^{n\times n} \times (C([0,T];\BR^{q}))^n \to (C([0,T];\BR^{q}))^n$$
 is characterized by a forward-backward coupled differential equation pair \cite[Prop. 1]{ShuangPeterMinyiTAC21}.
The equilibrium nodal cost is 
\[
\rho_i = J(z_i),  \quad z_i \in C([0,T]; \BR^{q}),  \quad i \in  [n]
\]
with the same $J(\cdot):C([0,T]; \BR^{q})\to \BR^+\cup 0$  for all nodes.
Hence it satisfies the definition of fixed-point centrality.
\end{proof}
We omit the details of the problem formulation of LQG Network Mean Field Games which requires significant space. Interested readers are referred to \cite[Sec. IV-B]{ShuangPeterMinyiTAC21}.  %
\subsection{Properties of Fixed-Point Centrality}
An \emph{automorphism} of a (directed or undirected) graph $\mathcal{G}(V, E)$ is  a permutation map $\pi:V \to V$ that satisfies 
$$
 (i, j) \in E \quad \text{if and only if} \quad  (\pi(i), \pi(j))\in E, ~ \forall i, j \in V.
$$

\begin{proposition}\label{prop:perminv-cent}
Any fixed-point centrality of a graph $\mathcal{G}(V,E)$ is permutation invariant with respect to any automorphism map of $\mathcal{G}(V,E)$. 
\end{proposition}
\begin{proof}
Let $A$ denote the adjacency matrix of $\mathcal{G}(V, E)$. 
	Let $A_\pi\triangleq P^\TRANS_\pi A P_\pi$ and $x_\pi \triangleq P_\pi x$, where $P_\pi$ is the permutation matrix corresponding to the permutation map $\pi:[n]\to [n]$.
	By the definition of an automorphism $\pi$, $A= A_{\pi}$ (that is, the adjacency matrix does not change),  and hence the fixed-point feature $x$ given by $x = f(A, x) $ satisfies that 
	\[
x_{\pi} = f(A_\pi, x_\pi) = f(A, x_\pi).
	\]
	{In the definition of fixed-point centrality, such fixed-point feature is assumed to be unique.  Then $x = x_\pi$.} That is, an automorphism does not change the fixed-point features and hence does not change the fixed-point centrality  $\rho=g(x)$. 
\end{proof}
	
	A \emph{vertex transitive graph} is a graph $\mathcal{G}$ satisfying that  for any given node pair $(i,j)$, there exists some automorphism map $\phi^{i,j} \in \Pi$ such that 
	$
	\phi^{i,j}(i) = j,
	$ where $\Pi$ denotes the set of permutation  mappings $\pi: [n]\to [n]$.
	  See \cite{weissteinVertexTransitive} for examples of vertex transitive graphs.

\begin{proposition}[Vertex Transitive Graphs]\label{eq:vertex-transitive}
	{All nodes of a vertex transitive graph share the same fixed-point centrality value, that is, any fixed-point centrality for a vertex transitive graph is permutation invariant.} 
\end{proposition}
\begin{proof}
Following Prop. \ref{prop:perminv-cent} and the definition of vertex transitive graphs, we obtain,  for each $i, j \in [n]$, there exists some $\phi^{i,j}\in \Pi $\footnote{There may be {one $\phi^{i,j}$ for each node pair $(i,j)$ instead of one $\phi$ for all node pairs}.}, such that the fixed-point features satisfy $x_i = x_{\phi^{i,j}(i)} = x_j$. 
This implies that $x_i = x_j$ for all $i,j \in [n]$. Finally, the permutation equivariance of  $g(\cdot)$ in \eqref{eq:fpc-def}   leads to  the desired result. %
\end{proof}

Properties in Prop. \ref{prop:perminv-cent} and Prop. \ref{eq:vertex-transitive} are general properties shared by all existing centralities that depend only on graph structures.
These properties may not hold in general for the outputs of graph neural network models (\cite{gori2005new,scarselli2009graph}). %
\subsection{Centrality Variations with Respect to Graph Variations}
Consider two graphs $\mathcal{G}(A)$ and $\mathcal{G}(B)$ with the same number of nodes. Let  $\rho_A$ be a fixed-point centrality  for $\mathcal{G}(A)$ and $\rho_B$ that of $\mathcal{G}(B)$  with the same function $f(\cdot,\cdot)$, that is, 
	\begin{equation} \label{eq:f-AB}
	\begin{aligned}
		& x_A = f(A, x_A),  \quad \rho_A = g(x_A), \\
		&x_B = f(B, x_B), \quad \rho_B = g(x_B),
	\end{aligned}
	\end{equation}
	where $S^n$ is specialized to $\BR^n$, and $f(\cdot, \cdot): {\BR}^{n\times n} \times \BR^n \to {\BR}^n$ and $g(\cdot): \BR^{n}\to \BR^n$. (The specialization of $S^n$ to $\BR^n$ is for the simplicity of presentation, and it can be relaxed to any normed vector space). 
In this section we study the conditions under which $\rho_A$ and $\rho_B$ are close and establish upper bounds of their differences. 

 Let $\mathcal{U}_f\subset \BR^n$ denote the set of feasible fixed-point features with $f(\cdot, \cdot)$ in \eqref{eq:f-AB}. Consider the following assumption.
\vspace{0.15cm}
\\
\textbf{Assumption (A1)}: 
(a)	There exists $L_1>0$ such that for all $x \in \mathcal{U}_f$,
	\begin{equation}
		\|f(A, x) - f(B,x)\|\leq L_1\|A-B\|_{\textup{op}},
	\end{equation}
	where the operator norm $
\|A\|_{\textup{op}} \triangleq\sup_{v\in \BR^n, v\neq 0}\frac{\|Av\|}{\|v\|}
$;\\
 (b)  For any matrix $A$ and for any $x \in \mathcal{U}_f$, there exists $L_0(A, x)\geq 0$ such that
	\begin{equation}\label{eq:relaxFPcen}
	\|f(A,x_A) -f(A, x)\|\leq L_0(A,x) \|x_A - x\|
\end{equation}
where $x_A = f(A, x_A)$;\\
(c) For the given matrix $A$,
\[
L_0(A) \triangleq\sup_{x \in \mathcal{U}_f} L_0(A, x)  <1;
\]
(d) There exists $L_g>0$ such that for all $x, v \in \mathcal{U}_f$, 
\[
 \|g(x) - g(v)\|\leq L_g \|x - v\|. 
\] 

 We call (A1)-(c) the \emph{Contraction Condition for Fixed-Point Centrality} for $\mathcal{G}(A)$; if, furthermore, $\mathcal{U}_f$ is complete under the chosen norm $\|\cdot\|$, it then gives the existence of a unique fixed-point feature for $f(A,\cdot)$ following Banach fixed-point theorem, and one can simply apply fixed-point iterations to identify such fixed-point feature  with the given graph $\mathcal{G}(A)$.

\begin{remark}[Different Choices of Norms]
We note that $\|\cdot\|$  can take any vector $\|\cdot\|_p$ norm, $1 \leq p\leq \infty$, as long as
the operator norm $\|\cdot \|_\textup{op}$ in (A1)-(a) is compatible with the chosen vector norm (that is,  $\|\cdot \|_\textup{op} = \sup_{v\in \BR^n, v\neq 0}\frac{\|Av\|_p}{\|v\|_p}$). 
\end{remark}

For Katz-Bonacich centrality, we choose $2$-norm and 
$L_0(A) = \alpha \|A\|_{\textup{2}}<1$ if $\alpha <\|A\|_{\textup{2}}^{-1}$. For PageRank centrality, we choose  $1$-norm  and 
$L_0(A) = \alpha  \| A^\TRANS \text{diag}(A^\TRANS\mathbf{1}_n)^{-1}\|_{\textup{1}}<1$ if $\alpha <  \|A^\TRANS \text{diag}(A^\TRANS\mathbf{1}_n)^{-1}\|_{\textup{1}}^{-1}$.

\begin{remark}[Fixed-Point Centrality in the Linear Case]
The condition $L_0(A)<1$ in (A1) is not satisfied  under $\|\cdot\|_2$ norm for the (normalized) eigencentrality, as $L_0(A)=1$ for eigencentrality. To establish the error bound, further spectral properties of the graphs are required (see e.g. \cite{avella2018centrality} via rotation analysis of eigenvectors by perturbations \cite{davis1970rotation}).
In general, for fixed-point centralities where $f(A,\cdot)$ is a linear function, one should establish the difference of two 1-dimensional subspaces characterized by $\text{span}(x_A)$ and $\text{span}(x_B)$. Such difference can be   characterized by the angular difference between the two subspaces as follows:
\[
{d}(x_A, x_B) = \left|\cos^{-1}{\left(\frac{\left|\langle x_A, x_B \rangle\right|}{\|x_A\|_2 \|x_B\|_2} \right)}\right|
\]
which is a specialization of a Grassmann distance (see  \cite{ye2016schubert}) to 1-dimensional subspaces (i.e. Grassmannian $\textup{Gr}(1,\BR^n)$). 
For characterizing such differences between $x_A$ and $x_B$ when $B$ differs from $A$ by a small perturbation, one may employ the error estimation results in \cite{davis1970rotation}.
\end{remark}

\begin{theorem}\label{thm:upperbound-centrality-good}
Under Assumption (A1)  for the fixed-point centrality \eqref{eq:f-AB}, the following holds
\begin{equation}
	\|\rho_A - \rho_B\| \leq  \frac{L_1 L_g}{1-L_0(A)} \|A -B\|_{\textup{op}}.
\end{equation}	
\end{theorem}
\begin{proof}
Following  the definition of the fixed-point centrality and Assumption (A1)(a)-(A1)(c), 
	\begin{equation*}
	\begin{aligned}
			\|x_A - &x_B\| = \|f(A,x_A) - f(B,x_B)\| \\ 
	& \leq \|f(A,x_A) - f(A,x_B)\|+ \|f(A,x_B) - f(B,x_B)\|\\
	& = L_0(A)\|x_A - x_B\| + L_1\|A- B\|_{\textup{op}}.
	\end{aligned}
\end{equation*}	
Hence subtracting $L_0(A)\|x_A - x_B\|$ and then dividing by $(1-L_0(A))$ on both sides yield
\begin{equation}
	\|\rho_A - \rho_B\| \leq  \frac{L_1}{1-L_0(A)} \|A -B\|_{\textup{op}}.
\end{equation}	
Then employing the condition in Assumption (A1)-(d) yields the desired result. 
\end{proof}

\begin{remark}
If  $f(A,\rho) $ does not depend on $\rho$, then the fixed-point centrality is trivial and the centrality variation upper bounds above should be treated differently. Such examples include degree, closeness and betweenness centralities.
\end{remark}
Centralities can be associated with probability distributions:  PageRank centrality is the steady state distribution of random walks on the graph of hyperlinks \cite{brin1998anatomy}, and degree centrality is used as the probability distribution for forming new connections \cite{barabasi1999emergence}. 
To (uniquely) associate the fixed-point centrality with a probability (mass function), we consider the following assumption. \vspace{0.15cm}
\\
\textbf{Assumption (A2)}: The fixed-point centralities are normalized  with nonnegative entries, that is, 
	\[
	\sum_{i\in [n]}  \rho_i=1, \quad \rho_i \geq 0,\quad  \forall i \in [n].
	\] 

Clearly, this implies $\|\rho\|_1\triangleq\sum_{i=1}^n|\rho_i|=1$.
\begin{remark}[Normalization of Centralities]
If a centrality $c$ does not satisfy the condition (A2) above, it can be normalized via 
$
{\rho}_i = \frac{c_i}{\sum_{i=1}^n c_i} , ~ i\in  [n].
$
 This normalization is  useful  to associate any centrality with a probability distribution. 
For instance, the degree centrality with normalization
$
\rho_i  = \frac{d_i}{\sum_{i=1}^n d_i}$ where $d_i =\sum_{j=1}^na_{ij}$, $i \in [n],
$ is used in scale-free network models \cite{barabasi1999emergence} to represent the probability of forming new connections.
In general, one can introduce a monotone function $\phi(\cdot): \BR \to \BR_{\geq 0}$, such that 
\[
{\rho}_i = \frac{\phi(c_i)}{\sum_{i=1}^n \phi(c_i)} , \quad i\in  [n].
\]
When $\phi(\cdot) = \exp(\cdot)$ (or $\phi(\cdot) = \exp(-\cdot)$), it is then specialized to the softmax function (or the Boltzmann distribution that maximizes an associated entropy). %
Such normalizations can be incorporated into the permutation equivariant mapping $g(\cdot)$ in the  definition of fixed-point centrality in \eqref{eq:fpc-def}.
\end{remark}

For a metric space $(\mathcal{X},d)$ and  
$p\geq 1$, let 
${\displaystyle P_{p}(\mathcal{X})}$ denote the set of all probability measures on $\mathcal{X}$ with finite 
$p${th} moment. 
The $p$-Wasserstein distance between two probability measures in ${\displaystyle P_{p}(\mathcal{X})}$ is defined as follows:
\begin{equation*}
		W_p(\rho_A, \rho_B)  = \bigg(\inf_{\gamma \in  \Gamma(\rho_A,\rho_B)}\int_{\mathcal{X}\times \mathcal{X}} d(x,y)^p d\gamma (x,y)\bigg)^{\frac1p}
\end{equation*}
where $\Gamma(\rho_A, \rho_B)$ denotes the set of probability measures on $\mathcal{X}\times \mathcal{X}$ with marginals $\rho_A$ and $\rho_B$.
\begin{proposition}\label{prop:w2-op}
Under Assumptions \textup{(A1)} and \textup{(A2)}, the following holds for the fixed-point centrality in \eqref{eq:f-AB}:
\begin{equation}
	W_p(\rho_A,\rho_B) \leq  \frac{L_1 L_g}{1-L_0(A)} \inf_{\pi \in \Pi}\|A^\pi -B\|_{\textup{op,p}},
\end{equation}	
where the matrix operator norm is $\|A\|_{\textup{op, p}} \triangleq \|A\|_p $. 
\end{proposition}
\begin{proof}
Recall from Theorem \ref{thm:upperbound-centrality-good} that
\begin{equation}
	\|\rho_A^{\pi^*}-\rho_B\|_p \leq  \frac{L_1 L_g}{1-L_0(A)} \|A^{\pi^*} -B\|_{\textup{op,p}},
\end{equation}
where $\pi^* =\arg\min_{\pi \in \Pi}\|A^\pi- B\|_{\text{op, p}}$. 	
Furthermore, one can verify that 
\[
W_p(\rho_A, \rho_B) \leq \|\rho_A^{\pi^*}-\rho_B\|_p,
\]
since $\pi^*$ is just a particular transport map. 
We obtain the desired result by combining the two inequalities above. 
\end{proof}

{When $p=2$, the operator norm  $\|A\|_{\textup{op},2}$ is the maximum singular value of $A$.}
Consider the matrix cut norm \cite{frieze1999quick} 
\begin{equation}\label{eq:cut-norm-matrix-no-scl}
	\|A\|_\Box \triangleq \max_{S\times T \subset [n]\times [n]}\Big| \sum_{i\in S, j\in T} a_{ij} \Big|, \quad A \in \BR^{n\times n}
\end{equation}
(without the scaling factor $\frac1{n^2}$ used in \cite[p.127]{lovasz2012large}). 
\begin{lemma}\label{lem:operator-cut-mat}
	The following  inequality holds for any symmetric matrix $A=[a_{ij}]$ with elements $|a_{ij}|\leq 1$:
\[
\|A\|_{\textup{op},2} \leq \sqrt{8 \|A\|_\Box}, \quad A \in \BR^{n\times n}.
\]
\end{lemma}
\begin{proof}
	For any symmetric matrix $A$,	the norms $\|A\|_{\textup{op, 2}}$ and $\|A\|_\Box$ scaled respectively by $\frac1n$ and $\frac{1}{n^2}$ correspond to those graphon norms $\|\FA\|_{\textup{op}}$ and $\|\FA\|_\Box$ in \cite{lovasz2012large} for the stepfunction graphon $\FA$ with uniform partitions associated with $A$, where, with an abuse of the notation, $\|\FA\|_\Box \triangleq \sup_{S, T\subset[0,1]}\left|\int_{S\times T} \FA(x,y)dxdy\right|$ and $\|\FA\|_{\textup{op}}\triangleq \sup_{ \Fv \in L^2([0,1])}{\frac{\|\FA \Fv\|_2}{\|\Fv\|_2}}$.
	Since  $\|\mathbf{\FA}\|_\textup{op} \leq \sqrt{8\|\mathbf{\FA}\|_\Box}$ holds for any graphon in  $\mathbf{\FA}\in \ESO$ (see \cite[Lem. E.6 and Eqn. (4.4)]{janson2010graphons} or \cite[Lem.~7]{avella2018centrality}), we obtained the desired result.
\end{proof}
%
 Replacing the operator norm by the cut norm in Prop.~\ref{prop:w2-op} via the inequality in Lemma \ref{lem:operator-cut-mat} yields the following result.
\begin{proposition} \label{prop:cut-matrix-bound}
Consider two symmetric matrices $A$ and $B$. Assume \textup{(A1)} and \textup{(A2)} for the fixed-point centrality \eqref{eq:f-AB} hold.  If $|a_{ij}|\leq 1$ and $|b_{ij}|\leq 1$ for all $i,j \in [n]$,  then 
\begin{equation}
	W_2(\rho_A,\rho_B) \leq  \frac{L_1 L_g}{1-L_0(A)}\sqrt{8 \delta_\Box(A,B)}
\end{equation}	
where $\delta_\Box(A,B) \triangleq \inf_{\pi \in \Pi}\|A^\pi -B\|_{\Box}$, $\|A\|_\Box \triangleq \max_{S\times T \subset [n]\times [n]}\Big| \sum_{i\in S, j\in T} a_{ij} \Big|$ and $\Pi$ denotes the set of all permutations from $[n]$ to $[n]$.
\end{proposition}


%
%
\section{Fixed-point Centrality for Infinite Networks}
Graphons are useful in characterizing and comparing graphs of different size and defining limits  of (deterministic or random) graph sequences. This section extends the fixed-point centralities to those for graphons. 
\subsection{Fixed-Point Centrality for Graphons}
{Let $\mathcal{W}_c$ denote the set of symmetric measurable functions $\FW:[0,1]^2\to [-c, c]$ with $c>0$.  Let $S^{[0,1]}$ denote the infinite Cartesian product of $S$ with the index set $[0,1]$. Let $d$ denotes the metric for $S^{[0,1]}$.} Similar to the finite graph case, a centrality for a graphon with the vertex set $[0,1]$ is defined as the mapping 
$\rho:[0,1]\to \BR_{\geq 0}$ which characterizes the ``importance" of nodes on the infinite network associated with the underlying graphon. 
%
\begin{definition}[Permutation Equivariant Operator]
	An operator $f(\cdot, \cdot): \mathcal{W}_c \times S^{[0,1]} \to S^{[0,1]}$ is \emph{permutation equivariant with respect to a measure preserving  {bijection}} $\pi : [0,1]  \to [0,1]$ if
\begin{equation}
	f( \FA,  \rho)^\pi 
  = f( \FA^\pi ,  \rho^\pi), ~~ \textcolor{black}{\forall \rho\in S^{[0,1]}, ~\forall \FA \in \ESC}
\end{equation}
where $\FA^\pi(\alpha, \beta)\triangleq \FA(\pi(\alpha), \pi(\beta))$ and $\rho^\pi(\alpha) \triangleq \rho(\pi(\alpha))$ for $\alpha,\beta \in[0,1]$.
 An operator $f(\cdot, \cdot): \mathcal{W}_c \times S^{[0,1]} \to S^{[0,1]}$ is \emph{permutation equivariant} if it is permutation equivariant with respect to all measure preserving  {bijections} $\pi : [0,1]  \to [0,1]$.
\end{definition}
\begin{definition}[Permutation Invariant Operator]
	An operator $f(\cdot, \cdot): \mathcal{W}_c \times S^{[0,1]} \to S^{[0,1]}$ is \emph{permutation invariant with respect to a measure preserving bijection} $\pi : [0,1]  \to [0,1]$ if
\begin{equation}
	f( \FA,  \rho) = f( \FA^\pi ,  \rho^\pi), ~ \textcolor{black}{\forall \rho\in S^{[0,1]}, ~\forall \FA \in \ESC,}
\end{equation}
where $\FA^\pi(\alpha, \beta)\triangleq \FA(\pi(\alpha), \pi(\beta))$ and $\rho^\pi(\alpha) \triangleq \rho(\pi(\alpha))$ for $\alpha,\beta \in[0,1]$. 
An operator $f(\cdot, \cdot): \mathcal{W}_c \times S^{[0,1]} \to S^{[0,1]}$ is \emph{permutation invariant} if it is {permutation invariant} with respect to {all measure preserving}  {bijections} $\pi : [0,1]  \to [0,1]$.
\end{definition}
Similarly, a mapping $g(\cdot): S^{[0,1]} \to S^{[0,1]}$ is permutation equivariant (resp. permutation invariant) if for all measure preserving  {bijections} $\pi : [0,1]  \to [0,1]$, 
$
g(\rho)^\pi = g(\rho^\pi) ~ (\text{resp.} ~ g(\rho)= g(\rho^\pi)).
$
%
\begin{definition}[Graphon Fixed-Point  Centrality] \label{def:graphon-fpc}
	A centrality $\rho:[0,1]\to \BR_{
	\geq 0}$  is a \emph{fixed-point centrality} for a graphon $\FA\in \ESC$ associated with the feature space $(S^{[0,1]},d)$ if there exists a permutation equivariant fixed-point mapping $f(\cdot, \cdot): \mathcal{W}_c \times S^{[0,1]} \to S^{[0,1]}$,  a permutation equivariant mapping $g(\cdot):   S^{[0,1]} \to \BR_{\geq 0}$, and a unique function $\Fx\in S^{[0,1]}$ under the metric $d$,   such that
\begin{equation}\label{eq:graphon-fpc}
\begin{aligned}
		& \Fx = f(\FA, \Fx),\\
		& \rho = g( \Fx), \quad \rho_\alpha \geq 0, ~~\alpha \in [0,1].
\end{aligned}
\end{equation}
\end{definition}
%
	We note that  the ``uniqueness'' of $\Fx$ in the definition above depends on the choice of $S^{[0,1]}$ and the underlying metric $d$, and it could mean an equivalent class of functions. For example, if we choose $S=\BR$ and let the set $\BR^{[0,1]}$ be endowed with $L^p([0,1])$ norm, then the unique $\Fx \in L^p([0,1])$ is interpreted as the equivalent class up to discrepancies on sets with Lebesgue measure zero. Another such example, 
	 similar to the finite graph case,  is that the  ``uniqueness'' of $\Fx$ when $f(\FA, \cdot)$ is a linear mapping shall be interpreted as the unique subspace spanned by $\Fx$ (see Remark \ref{eq:linear-case}).
%

\begin{proposition}
Graphon eigencentrality, graphon Katz-Bonacich centrality, and  graphon PageRank centrality  are graphon fixed-point centralities. 
\end{proposition}
\begin{proposition}
Any eigenfunction of a graphon operator from $L^2([0,1])$ to $L^2([0,1])$ corresponding to a non-zero simple eigenvalue is a graphon fixed-point centrality.
\end{proposition}

\begin{proposition}
The equilibrium nodal cost of LQG Graphon Mean Field Games \textup{\cite[Sec. IV-C]{ShuangPeterMinyiTAC21}} with homogenous initial conditions, if the unique equilibrium exists, is a graphon fixed-point centrality. 
\end{proposition}
In LQG Graphon Mean Field Games \cite[Sec. IV-C]{ShuangPeterMinyiTAC21}, the product set $S^{[0,1]}$ is specialized to
$C([0,T];(L^2([0,1]))^q)$, where $q\geq 1$ is the dimension of the local state of agents. 
Proofs of these propositions follow similar arguments as those in the finite network case, and hence are omitted.
%
\subsection{Centrality Variations with Respect to Graphon Variations}
Consider two graphons $\FA$ and $\FB$ in $\ESC$, and let  $\rho_\FA$ and  $\rho_\FB$ be respectively their fixed-point centralities  as in \eqref{eq:graphon-fpc}, that is, 
	\begin{equation} \label{eq:graphon-f-AB}
	\begin{aligned}
		& \Fx_\FA = f(\FA, \Fx_\FA),  \quad \rho_A = g(\Fx_\FA), \\
		&\Fx_\FB = f(\FB, \Fx_\FB), \quad \rho_B = g(\Fx_\FB),
	\end{aligned}
	\end{equation}
	where  the feature space $S^{[0,1]}$ is specialized to $L^p([0,1])$ with $p\geq 1$, and the operators $f(\cdot,\cdot)$ and $g(\cdot)$ are specialized to $f(\cdot, \cdot): \mathcal{W}_c \times L^p([0,1]) \to L^p([0,1])$ and $g(\cdot):   L^p([0,1]) \to L^p([0,1])$.
 Let $\mathcal{U}_f\subset L^p([0,1])$ denote the set of feasible fixed-point features associated with $f(\cdot, \cdot)$ in \eqref{eq:graphon-f-AB}. Consider the following assumption.\vspace{0.15cm}
 \\
\textbf{Assumption (A3)}:
(a)	There exists $L_1>0$ such that for all $\Fx \in \mathcal{U}_f$,
	\begin{equation}
		\|f(\FA, \Fx) - f(\FB,\Fx)\|\leq L_1\|\FA-\FB\|_{\textup{op}},
	\end{equation}
	where the operator norm $
\|\FA\|_{\textup{op}} \triangleq\sup_{\Fx\neq 0, \|\Fx\|<\infty }\frac{\|\FA\Fx\|}{\|\Fx\|}
$.\\
 (b)  For any graphon $\FA \in \ESC$ and $\Fx \in \mathcal{U}_f$, there exists $L_0(\FA, \Fx)\geq 0$ such that
	\begin{equation}\label{eq:relaxFPcen}
	\|f(\FA,\Fx_\FA) -f(\FA, \Fx)\|\leq L_0(\FA,\Fx) \|\Fx_\FA - \Fx\|
\end{equation}
where $\Fx_A = f(A, \Fx_A)$.\\
(c) For the given graphon $\FA$,
\[
L_0(\FA) \triangleq\sup_{\Fx \in \mathcal{U}_f} L_0(\FA, \Fx)  <1.
\] \\
(d) There exists $L_g>0$ such that for all $\Fx, \Fv \in \mathcal{U}_f$, 
\[
 \|g(\Fx) - g(\Fv)\|\leq L_g \|\Fx - \Fv\|. 
\] 
\begin{theorem}\label{thm:graphon-upperbound-centrality-good}
Under Assumption \textup{(A3)}  for the graphon fixed-point centrality \eqref{eq:graphon-f-AB}, the following holds
\begin{equation}
	\|\rho_\FA - \rho_\FB\| \leq  \frac{L_1 L_g}{1-L_0(\FA)} \|\FA -\FB\|_{\textup{op}}.
\end{equation}	
\end{theorem}
The proof essentially follows the same lines of arguments as those for Theorem \ref{thm:upperbound-centrality-good}. 
\vspace{0.15cm}\\
\textbf{Assumption (A4)}: The graphon fixed-point centrality $\rho$  satisfies
\begin{equation}
	\int_{[0,1]} \rho_\alpha d\alpha  = 1 \quad \text{and} \quad \rho_\alpha \geq 0,
\end{equation}
that is, $\|\rho\|_1 \triangleq \int_{[0,1]} |\rho_\alpha| d \alpha =1$. 
\begin{proposition}\label{prop:centrality-graphon-norm}
Under Assumptions \textup{(A3)} and \textup{(A4)}, the following holds for the fixed-point centrality in \eqref{eq:graphon-f-AB}:
\begin{equation}
	W_p(\rho_\FA,\rho_\FB) \leq  \frac{L_1 L_g}{1-L_0(\FA)} \inf_{\phi \in \Phi}\|\FA^\phi -\FB\|_{\textup{op,p}},
\end{equation}	
where $\Phi$ denotes the set of all measure preserving bijections from $[0,1]$ to $[0,1]$ and
 the operator norm is $\|\FA\|_{\textup{op, p}} \triangleq \sup_{\Fx\neq 0, \Fx \in L^p([0,1])}\frac{\|\FA\Fx\|_p}{\|\Fx\|_p} $. 
\end{proposition}

\begin{proposition}\label{prop:centrality-graphon-metric}
Consider two graphons $\FA$ and $\FB$ in $\ESO$. Assume \textup{(A3)} and \textup{(A4)} for the graphon fixed-point centrality \eqref{eq:graphon-f-AB} hold.  Then the following holds
\begin{equation}
	W_2(\rho_\FA,\rho_\FB) \leq  \frac{L_1 L_g}{1-L_0(\FA)}\sqrt{8 \delta_\Box(\FA,\FB)}.
\end{equation}	
where $\delta_\Box(\FA,\FB)\triangleq \inf_{\phi \in \Phi}\|\FA^\phi-\FB\|_\Box$ and $\|\FA\|_\Box \triangleq \sup_{S, T\subset[0,1]}\left|\int_{S\times T} \FA(x,y)dxdy\right|$, and $\Phi$ denotes the set of all measure preserving bijections $\phi:[0,1]\to[0,1]$. 
\end{proposition}
Proofs follow similar arguments as those in Prop.\ref{prop:w2-op} and  \ref{prop:cut-matrix-bound}.
\begin{remark}
Any finite undirected graphs can be represented by stepfunction graphons \cite[Chp.7.1]{lovasz2012large} and hence the characterization of  centrality variations applies to finite graphs as well. Moreover,  finite graphs of different size can be compared via their graphon representations as well as the associated fixed-point centralities. Thus, for the undirected graph case, the results above in Prop. \ref{prop:centrality-graphon-norm} and \ref{prop:centrality-graphon-metric} generalize those in Prop. \ref{prop:w2-op} and \ref{prop:cut-matrix-bound}.%
\end{remark}
\section{Conclusion}
The notion of the fixed-point centrality proposed in the current paper is useful in at least the following ways: (a)  it helps identify properties for a large family of centralities and apply similar analysis techniques (e.g. in studying changes of centralities with respect to graph perturbations); (b) the well-established theoretical and numerical results of fixed-point analysis can be readily employed for such centralities; (c) learning and training methods can be readily applied to approximate fixed-point centralities due to its close connection with graph neural networks (\cite{gori2005new,scarselli2009graph}).

 The connection of fixed-point centrality with LQG  mean field games on networks suggests collective multi-agent learning of centralities from the equilibrium cost for (dynamic or repeated) game problems. %
Fixed-point centralities are also related to certain Markov decision processes if each state is viewed as a node and the value function (typically characterized by the fixed-point of the Bellman operator) is then a mapping from the vertex set to non-negative real number.  Details will be discussed in future extensions.

 The representation of sparse graph sequences and  limits requires extra concepts (e.g. graphings for bounded degree graphs \cite{lovasz2012large} and  $L^p$ graphon for sparse  $W$-random graphs \cite{borgs2018p}). Future extensions should  formulate fixed-point centralities for sparse graph limit models.
Other important future directions include: 
(a) improving upper bounds for centrality variations by exploring further properties of the permutation equivariant mappings;  (b) axoimatizing fixed-point centralities via extra properties of $f(\cdot, \cdot)$, $g(\cdot)$, and the feature space $S$ similar to that in \cite{bloch2016centrality};
(c) analyzing the change of the ranking properties of fixed-point centralities with respect to modification on networks;
(d) exploring variational analysis of fixed-point centralities where the underlying graphs are characterized by vertexon-graphons \cite{PeterVertexon}.

\bibliographystyle{IEEEtran}
\bibliography{mybib}
\end{document}